\documentclass[11pt, letterpaper]{amsart}
\usepackage{graphicx, amssymb}

\newtheorem{thm}{Theorem}[section]

\newtheorem{lem}[thm]{Lemma}
\newtheorem{prop}[thm]{Proposition}
\theoremstyle{definition}
\newtheorem{defn}[thm]{Definition}

\theoremstyle{remark}
\newtheorem{rem}[thm]{Remark}

\numberwithin{equation}{section}

\newcommand{\set}[1]{\left\{#1\right\}}
\newcommand{\Real}{\mathbb R}

\newcommand{\Natural}{\mathbb N}

\newcommand{\B}{\mathcal{B}}
\newcommand{\such}{{\ | \ }}
\newcommand{\limn}{\lim_{n \to \infty}}

\newcommand{\dfn}{\, := \,}

\newcommand{\prob}{\mathbb{P}}

\newcommand{\plim}{\prob \text{-} \lim}
\newcommand{\plimn}{\plim_{n \to \infty}}

\newcommand{\qprob}{\mathbb{Q}}
\newcommand{\expec}{\mathbb{E}}

\newcommand{\Lb}{\mathbb{L}}
\newcommand{\lz}{\Lb^0}
\newcommand{\lzp}{\lz_+}
\newcommand{\lzpp}{\lz_{++}}
\newcommand{\lzcad}{c\`ad in probability}
\newcommand{\lzcadlag}{c\`adl\`ag in probability}
\newcommand{\pbdd}{bounded in probability}
\newcommand{\F}{\mathcal{F}}
\newcommand{\G}{\mathcal{G}}
\newcommand{\Hcal}{\mathcal{H}}
\newcommand{\cadlag}{c\`adl\`ag}

\newcommand{\num}{num\'eraire}
\newcommand{\X}{\mathcal{X}}

\newcommand{\nin}{n \in \Natural}

\newcommand{\tX}{\widetilde{X}}
\newcommand{\hX}{\widehat{X}}
\newcommand{\hY}{\widehat{Y}}
\newcommand{\oX}{\overline{\X}}
\newcommand{\ox}{\overline{X}}
\newcommand{\oy}{\overline{Y}}

\newcommand{\fhat}{\widehat{f}}

\newcommand{\hti}{\widetilde{h}}

\newcommand{\tf}{\widetilde{f}}

\newcommand{\pare}[1]{\left(#1\right)}
\newcommand{\bra}[1]{\left[#1\right]}
\newcommand{\dbra}[1]{[\kern-0.15em[ #1 ]\kern-0.15em]}
\newcommand{\dbraco}[1]{[\kern-0.15em[ #1 [\kern-0.15em[}
\newcommand{\dbraoc}[1]{]\kern-0.15em] #1 ]\kern-0.15em]}
\newcommand{\C}{\mathcal{C}}

\newcommand{\bF}{\mathbf{F}}

\newcommand{\indic}{\mathbb{I}}

\newcommand{\absco}{{<\kern-0.53em<}}

\newcommand{\oC}{\overline{\C}}

\newcommand{\dda}{\downarrow \downarrow}
\newcommand{\dua}{\uparrow \uparrow}

\newcommand{\otcT}{]t, T]}

\begin{document}

\title[Generalized supermartingale deflators under limited information]{Generalized supermartingale deflators under limited information}%
\author{Constantinos Kardaras}%
\address{Constantinos Kardaras, Mathematics and Statistics Department, Boston University, 111 Cummington Street, Boston, MA 02215, USA.}%
\email{kardaras@bu.edu}%

\thanks{The author acknowledges partial support by the National Science Foundation under award number DMS-0908461. Furthermore, the author would like to express his gratitude to two anonymous referees and the involved Associate Editor of ``Mathematical Finance'' for very valuable input that improved the presentation of the paper.
}%

\date{\today}%
\begin{abstract}
We undertake a study of markets from the perspective of a financial agent with limited access to information. The set of wealth processes available to the agent is structured with reasonable economic properties, instead of the usual practice of taking it to consist of stochastic integrals against a semimartingale integrator. We obtain the equivalence of the boundedness in probability of the set of terminal wealth outcomes (which in turn is equivalent to the weak market viability condition of absence of arbitrage of the first kind) with the existence of at least one strictly positive deflator that makes the deflated wealth processes have a generalized supermartingale property. 
\end{abstract}

\maketitle

\setcounter{section}{-1}

\section{Introduction}

An almost universal assumption in the literature of financial mathematics is that prices of traded assets, and as a byproduct wealth processes resulting from trading, are directly observable from an acting agent in the market. In mathematical terminology, one postulates that wealth processes are adapted with respect to the agent's filtration. In practice, however,  it is not always reasonable to assume the agent's information flow is large enough to satisfy the previous requirement. This can model, for example, cases where information arrives to the agent with a delay, or in limited form. Additionally, it can model circumstances where there is lag between the decisions of the agent and their implementation; in that case, prices at the moment when the act is implemented are unknown at the moment when the decision is made.

The purpose of this work is to study market viability in scenarios like the ones described above. All wealth processes available to an agent with some fixed initial capital are modeled via an abstract set $\X$. The agent possesses some information stream under which the wealth processes are not necessarily adapted. The aforementioned set $\X$ is endowed with a reasonable economical structure, but it is \emph{not} assumed to be generated by results of integrals against semimartingales. (To begin with, such an assumption would not make sense in our ``limited information'' set-up. Furthermore, the freedom we are allowing in the definition of a wealth-process set naturally allows for situations where an infinite number of underlying assets are available for trading, as is for example the case in the theoretical modeling of bond markets.) The main result of this paper establishes the equivalence between the boundedness in probability of $\set{X_T \such X \in \X}$, where $T$ denotes a finite time-horizon, and the existence of at least one strictly positive process $Y$ such that all deflated processes $Y X$, where $X$ ranges in $\X$, have some ``generalized supermartingale'' property under the agent's filtration. Boundedness in probability of the set of terminal wealth processes has been discussed in great detail in \cite{MR1304434}, \cite{MR1647282} and \cite{MR2335830}; it is actually equivalent to a weak \emph{market viability} condition, namely \emph{absence of arbitrages of the first kind}, discussed in \cite{Kard10}. As it turns out, the correct description of a strictly positive deflator $Y$ to be used in the aforementioned equivalence involves a ``multiplicative'' generalization of supermartingales. In the full-information case, meaning that wealth processes in $\X$ are observable to the agent, the latter generalization exactly reduces to the familiar supermartingale property. 

Literature dealing with viability of markets where agents have limited information is scarce. To the best of the author's knowledge, a treatment of this problem in a \emph{continuous}-time setting has not appeared before. For discrete-time models, it was shown in in \cite{MR2276898} that the classical ``No Arbitrage'' condition is equivalent to the existence of a probability $\qprob$, equivalent to $\prob$, such that the optional projection (on the agent's filtration) under $\qprob$  of the discounted asset-prices  are $\qprob$-martingales. In the latter paper, the authors argue that questions regarding viability for continuous-time models can be posed even for processes that are not semimartingales. It is then not hard to understand why such line of research does not seem very promising in continuous time: all the rich machinery of semimartingale theory cannot be directly used, since the processes involved might fail to be adapted with respect to the agent's filtration. Indeed, a portion of the work carried out in this paper deals with establishing appropriate generalizations of well-known results, such as Doob's nonnegative supermartingale convergence theorem, in order to achieve the goal of proving the main result. In this sense, this paper also contributes to the general theory of stochastic processes.

\smallskip

The structure of the paper is simple: in Section \ref{sec: results} the result is stated, while Section \ref{sec: proof of main thm} contains its somewhat lengthy and technical proof.

\section{The Result} \label{sec: results}

\subsection{Probabilistic notation and definitions}

All stochastic elements in the sequel are defined on a probability space $(\Omega,  \G, \, \prob)$, where $\G$ is a $\sigma$-field over $\Omega$ and $\prob$ is a probability on $(\Omega, \G)$. Fix some $T \in \Real_{+}$ that models the end of financial activity. We consider a right-continuous filtration $\bF = (\F_t)_{t \in [0, T]}$ such that $\F_t \subseteq \G$ holds for all $t \in [0, T]$ and $\F_0$ is trivial modulo $\prob$. We assume that $\G$ is $\prob$-complete and all $\prob$-null sets of $\G$ are contained in $\F_0$; in other words, the stochastic basis satisfies the usual hypotheses.

We stress that the stochastic processes that will be considered in what follows are \emph{not} assumed to be $\bF$-adapted; by a ``stochastic process $X$'' we simply mean a collection $(X_t)_{t \in [0, T]}$ such that, for each $t \in [0, T]$, $X_t$ is a $\G$-measurable random variable.

By $\lzp$ we shall be denoting the set of all equivalence classes (modulo $\prob$) of nonnegative, $\G$-measurable random variables, endowed with the metric topology of convergence in $\prob$-measure. (Note that we shall not differentiate between random variables and the equivalence class in $\lzp$ they generate.) Furthermore, we shall use $\lzpp$ to denote the set of $f \in \lzp$ such that $\prob[f > 0] = 1$.

\begin{defn} \label{dfn: processes}
A stochastic process $X$ will be called \textsl{nonnegative} if $X_t \in \lzp$ for all $t \in [0, T]$; $X$ will be called \textsl{strictly positive} if $X_t \in \lzpp$ for all $t \in [0, T]$. A nonnegative stochastic process $X$ will be called \textsl{\lzcad} if the mapping $[0, T] \ni t \mapsto X_t \in \lzp$ is right-continuous. Further, a nonnegative process $X$ will be called \textsl{\lzcadlag} if the mapping $[0, T] \ni t \mapsto X_t \in \lzp$ is right-continuous and admits left-hand limits.
\end{defn}

The notions of process-continuity in Definition \ref{dfn: processes} are weaker than the corresponding ``c\`ad'' \ and ``\cadlag'' \ notions referring to the paths of a process.

\subsection{Generalized supermartingales} \label{subsuc: gen supermarts}

We now introduce a ``supermartingale'' property with respect to $\bF$ for nonnegative processes, when these processes are not necessarily $\bF$-adapted.

\begin{defn} \label{dfn: gen supermarts}
A \emph{nonnegative} stochastic process $Z$ will be called a \textsl{generalized supermartingale with respect to $\bF$} if
$\expec[Z_t / Z_s \such \F_s] \leq 1$ holds whenever $s \in [0, T]$ and $t \in [s, T]$.
\end{defn}

In the context of Definition \ref{dfn: gen supermarts}, the event $\set{Z_s = 0} \in \G$ might not be $\prob$-null for $s \in [0, T]$; therefore, one should be careful in defining $Z_t / Z_s$ on $\set{Z_s = 0}$. We use the following conventions: on $\set{Z_s = 0, \, Z_t > 0}$ we set  $Z_t / Z_s = \infty$, while on $\set{Z_s = 0, \, Z_t = 0}$ we set  $Z_t / Z_s = 1$. In particular, if $Z$ is a nonnegative generalized supermartingale with respect to $\bF$, then $\prob[Z_s = 0, \, Z_t > 0] = 0$ holds whenever $s \in [0, T]$ and $t \in [s, T]$.


If a nonnegative process $Z$ is $\bF$-adapted, it is straightforward to check (using our division conventions) that $Z$ is a generalized supermartingale with respect to $\bF$ if and only if $\expec[Z_t \such \F_s] \leq Z_s$ holds whenever $s \in [0, T]$ and $t \in [s, T]$; in other words, we retrieve the classical definition of nonnegative supermartingales.

\subsection{The equivalence result}

We are ready to state the main result of the paper, which connects the boundedness in probability of the terminal values of a set of wealth processes to the existence of a strictly positive generalized supermartingale deflator. Theorem \ref{thm: main dynamic} below, whose proof is given in Section \ref{sec: proof of main thm}, refines and widens the scope of previous findings obtained in the ``full information'' case, like the ones in \cite{MR2284490} and \cite{MR2335830}.

\begin{thm} \label{thm: main dynamic}
Let $\X$ be a set of stochastic processes such that:
\begin{enumerate}
\item[(a)] Each $X \in \X$ is nonnegative and \lzcad, and satisfies $X_0 = 1$.
\item[(b)] There exists a strictly positive process $\ox \in \X$.
\item[(c)] $\X$ is \textsl{convex}: $\pare{(1 - \alpha) X + \alpha X'} \in \X$ holds for any $X \in \X$, $X' \in \X$, and $\alpha \in [0,1]$.
\item[(d)] $\X$ has the following \textsl{switching property}: for all $\tau \in [0, T]$ and $A \in \F_\tau$, all $X \in \X$, and all strictly positive $X' \in \X$, the process
\[
\indic_{\Omega \setminus A} X_\cdot + \indic_A \frac{X'_{\tau \vee \cdot}}{X'_\tau} X_{\tau \wedge \cdot} = \left\{
           \begin{array}{ll}
             X_t (\omega), & \hbox{if $t \in [0, \tau[$, or $\omega \notin A$;} \\
             \pare{X_{\tau} (\omega) / X_{\tau}' (\omega)} X_t'(\omega) , & \hbox{if $t \in [\tau, T]$ and $\omega \in A$}
           \end{array}
         \right.
\]
is also an element of $\X$.
\end{enumerate}
Then, the following statements are equivalent:
\begin{enumerate}
  \item The set $\set{X_T \such X \in \X}$ is bounded in probability: $\lim_{\ell \to \infty} \sup_{X \in \X} \prob[X_T > \ell] = 0$.
  \item There exists a \lzcadlag \ and strictly positive process $Y$ such that $Y X$ is a generalized supermartingale with respect to $\bF$ for all $X \in \X$.
\end{enumerate}
Under any of the above equivalent conditions, each $X \in \X$ is \lzcadlag.

\smallskip

If $\X$ is such that \emph{(a)} through \emph{(d)} are satisfied and furthermore $\set{X_T \such X \in \X}$ is closed in probability, conditions \emph{(1)} and \emph{(2)} above are also equivalent to:
\begin{enumerate}
  \item[($3$)] There exists a strictly positive wealth process $\hX \in \X$ such that $X / \hX$ is a generalized supermartingale with respect to $\bF$ for all $X \in \X$.
\end{enumerate}
\end{thm}

\subsection{Remarks on Theorem \ref{thm: main dynamic}}

We continue by discussing some topics that are related to the statement of Theorem \ref{thm: main dynamic}.

\subsubsection{Financial interpretation of the set $\X$} \label{subsubsec: financial interpretation}
A set $\X$ that satisfies (a) through (d) in the statement of Theorem \ref{thm: main dynamic} can be thought as modeling the wealth processes that are available to some agent in a financial market. Condition  (a) states that the initial capital of the agent is normalized to unit, and that wealth processes satisfy an extremely mild ``regularity'' requirement. Condition (b) states that one can find a wealth process in $\ox \in \X$ that can be used as a ``baseline'' to denominate all other wealths --- for this reason, it has to be strictly positive. (Usually, $\ox$ is taken to be the wealth process generated by the bank account.) Note that if we choose to actually denominate all wealths in units of $\ox$, in other words if we replace $\X$ by $\oX \dfn \set{X /\ox \such X \in \X}$, then properties (a) through (d) of Theorem \ref{thm: main dynamic} still hold for the new wealth-process set $\oX$ with (b) actually strengthened to $1 \in \oX$. Further, note that a process $Y$ satisfies statement (2) of Theorem \ref{thm: main dynamic} if and only if the process $\oy \dfn Y \ox$ satisfies statement (2) of Theorem \ref{thm: main dynamic} with $\oX$ replacing $\X$. This simple ``change of \num'' trick helps reduce the proof of Theorem \ref{thm: main dynamic} to the case where property (b) is strengthened to $1 \in \X$. Moving ahead, it is intuitively clear why the convexity property (c) should hold: if an agent can invest in two wealth processes $X \in \X$ and $X' \in \X$, the agent should be free to allocate at time $t = 0$ a fraction $\alpha \in [0,1]$ of the unit initial capital to wealth $X'$ and the remaining fraction to the wealth $X$. The switching property (d) has the following economic interpretation: if an agent can invest in two wealth streams $X \in \X$ and $X' \in \X$, where the latter process is assumed strictly positive, we should then allow for the possibility that, starting with the wealth process $X$, at time $\tau$ the agent decides to either switch to the wealth process $X'$, which happens on $A \in \F_\tau$, or keep investing according to $X$, on the event $\Omega \setminus A$. Note that it is exactly condition (d) which reflects that the information flow available to the agent is $\bF$.

\subsubsection{Market viability}

An \textsl{arbitrage of the first kind} in the market is a random variable $\xi \in \lzp$ with $\prob[\xi > 0] > 0$ such that for all $x > 0$ there exists $X \in \X$ (which may depend on $x$) which satisfies $\prob [x X_T \geq \xi] = 1$. We shall say that condition NA$_1$ holds if there is no arbitrage of the first kind in the market.
In par with the financial interpretation of the set $\X$ given in \S \ref{subsubsec: financial interpretation} above, the set $x \X = \set{x X \such X \in \X}$ corresponds to all attainable wealth processes starting from initial capital $x > 0$. Therefore, in words, condition NA$_1$ fails if and only if no matter how minute the initial capital is, an investor can invest in a way that certainly results at time $T$ in at least a predetermined non-zero (on a set of strictly positive probability) amount.  According to \cite[Proposition 1]{Kard10}, boundedness in probability of $\set{X_T \such X \in \X}$ is equivalent to condition NA$_1$. (Although the set-up in \cite{Kard10} is different, the proof of Proposition 1 from the latter paper can be copied \emph{mutatis mutandis} for the present situation.) The aforementioned equivalence clarifies the financial relevance of Theorem \ref{thm: main dynamic}.

\subsubsection{The num\'eraire in $\X$} \label{subsubsec: numeraire}

When $\set{X_T \such X \in \X}$ is bounded and closed in probability, it is natural to call a process $\hX$ that satisfies condition $(3)$ of Theorem \ref{thm: main dynamic} above \textsl{the \num} in $\X$, which generalizes the definition for the full-information case (see \cite{LONG}, \cite{MR1849424}, \cite{MR2335830}). Note that the \num \ in $\X$, if it exists, is unique up to modification. Indeed, suppose that both strictly positive processes $\hX \in \X$ and $\hX' \in \X$ are such that $\hX / \hX'$ and $\hX' / \hX$ are generalized supermartingales with respect to $\bF$. In particular, $\expec [\hX'_t / \hX_t] \leq 1$ and $\expec [\hX_t / \hX'_t] \leq 1$ should hold simultaneously for all $t \in [0, T]$. Jensen's inequality gives that $\prob [\hX_t = \hX'_t] = 1$ for all $t \in [0, T]$.


\subsubsection{Adaptedness of the strictly positive generalized supermartingale deflator}

As a careful inspection of the proof of Theorem \ref{thm: main dynamic} in Section \ref{sec: proof of main thm} reveals, the strictly positive generalized supermartingale deflator $Y$ that satisfies condition (2) of Theorem \ref{thm: main dynamic} can be chosen to be adapted with respect to the usual augmentation of the filtration that makes all the wealth processes in $\X$ adapted. In particular, if the processes in $\X$ are $\bF$-adapted, $Y$ can be chosen to be $\bF$-adapted.

\section{Proof of Theorem \ref{thm: main dynamic}} \label{sec: proof of main thm}

We start by mentioning (without proof) a special case of \cite[Lemma A1.1]{MR1304434}, which will be used constantly throughout the proof of Theorem \ref{thm: main dynamic}. Recall that a set $\B \subseteq \lzp$ is called \pbdd \ if $\lim_{\ell \to \infty} \sup_{f \in \B} \prob[f > \ell] = 0$.

\begin{lem} \label{lem: A1.1}
Let $(f^n)_{\nin}$ be an $\lzp$-valued sequence and define $\C^n$ as the convex hull of the set $\set{f^n, f^{n+1}, \ldots}$, for each $\nin$. Assume that $\C^1$ is bounded in probability. Then, there exists $g \in \lzp$ and a sequence $(g^n)_{\nin}$ such that $g^n \in \C^n$ for all $\nin$ and $\plimn g^n = g$.
\end{lem}

We now state and prove a ``static'' version of Theorem \ref{thm: main dynamic}.

\begin{thm} \label{thm: main static}
Let $\C \subseteq \Lb^0_+$ with $\C \cap \lz_{++} \neq \emptyset$. Assume that $\C$ is convex and closed in probability. Then, the following statements are equivalent:
\begin{enumerate}
  \item $\C$ is \pbdd.
  \item There exists $g \in \lzpp$ such that $\expec[g f] \leq 1$ holds for all $f \in \C$.
  \item There exists $\fhat \in \C \cap \lz_{++}$ such that $\expec[f / \fhat] \leq 1$ holds for all $f \in \C$.
\end{enumerate}
\end{thm}

\begin{proof} Implication $(3) \Rightarrow (2)$ trivially follows by setting $g \dfn 1 / \fhat$. Further, assume (2)  and fix $g \in \lzpp$ such that $\expec[g f] \leq 1$ holds for all $f \in \C$. For all $\ell \in \Real_+$ and $f \in \C$, $\ell \, \prob[f g > \ell] \leq \expec[f g] \leq 1$. Therefore, $\lim_{\ell \to \infty} \sup_{f \in \C} \prob[f g > \ell] \leq \limsup_{\ell \to \infty} (1 / \ell) = 0$, i.e., $\set{f g \such f \in \C}$ is bounded in probability. Since $g \in \Lb^0_{++}$, $\C$ is also bounded in probability, i.e., condition (1) holds.

\smallskip

We now discuss the more difficult implication $(1) \Rightarrow (3)$, which is the content of \cite[Theorem 1.1(4)]{Kar09}. Since $\C \cap \lzpp \neq \emptyset$, we may assume that $1 \in \C$. Indeed, otherwise, we consider $\widetilde{\C} \dfn (1 / g) \C$ for some $g \in \C \cap \Lb^0_{++}$. Then, $1 \in \widetilde{\C}$ and $\widetilde{\C}$ is still convex, closed and bounded in probability. Furthermore, if $\expec[f / \widetilde{f}] \leq 1$ holds for all $f \in \widetilde{\C}$, then, with $\fhat \dfn g \widetilde{f}$, $\expec[f / \fhat] \leq 1$ holds for all $f \in \C$. Therefore, in the sequel we assume that $1 \in \C$. We claim that we can further assume without loss of generality that $\C$ is solid. Indeed, let $\C'$ be the \emph{solid hull} of $\C$, i.e., $\C' \dfn \set{f \in \lz_{+} \such f \leq h \text{ holds for some } h \in \C}$. It is straightforward that $1 \in \C'$, as well as that $\C'$ is still convex and bounded in probability. It is also true that $\C'$ is still closed in probability. (To see the last fact, pick a $\C'$-valued sequence  $(f^n)_{n \in \Natural}$ that converges $\prob$-a.s. to $f \in \lz_{+}$. Let $(h^n)_{n \in \Natural}$ be a $\C$-valued sequence with $f^n \leq h^n$ for all $n \in \Natural$. By Lemma \ref{lem: A1.1}, we can extract a sequence $(\hti^n)_{n \in \Natural}$ such that, for each $n \in \Natural$, $\hti^n$ is a convex combination of $h^n, h^{n+1}, \ldots$, and such that $h \dfn \plimn \hti^n$ exists. Of course, $h \in \C$ and it is easy to see that $f \leq h$. We then conclude that $f \in \C'$.) Suppose that there exists $\fhat \in \C'$ such that $\expec[f / \fhat] \leq 1$ holds for all $f \in \C'$. Then, $\fhat \in \C$ (since $\fhat$ has to be a \emph{maximal} element of $\C'$ with respect to the order structure of $\lz$), and $\expec[f / \fhat] \leq 1$ holds for all $f \in \C$ (simply because $\C \subseteq \C'$). To recapitulate, in the course of the proof of implication $(1) \Rightarrow (3)$, we shall be assuming without loss of generality that $\C \subseteq \lz_{+}$ is solid, convex, closed and bounded in probability, as well as that $1 \in \C$.

For all $n \in \Natural$, let $\C^n \dfn \set{f \in \C \such f \leq n}$, which is convex, closed and bounded in probability and satisfies $\C^n \subseteq \C$. Consider the following optimization problem:
\begin{equation} \label{eq: log-optimal prob}
\text{find } f_*^n \in \C^n \text{ such that } \expec[\log(f_*^n)] = \sup_{f \in \C^n} \expec[\log(f)].
\end{equation}
The fact that $1 \in \C^n$ implies that the value of the above problem is not $- \infty$. Further, since $f \leq n$ for all $f \in \C^n$, one can use of Lemma \ref{lem: A1.1} in conjunction with the inverse Fatou's lemma and obtain the existence of the optimizer $f_*^n$ of \eqref{eq: log-optimal prob}. For all $f \in \C^n$ and $\epsilon \in \, ]0, 1/2]$, one has
\begin{equation} \label{eq: finite diff}
\expec \bra{\Delta_\epsilon (f \such f_*^n ) } \leq 0, \text{ where } \Delta_\epsilon (f \such f_*^n ) \dfn \frac{\log \pare{(1 - \epsilon) f_*^n + \epsilon f} - \log \pare{f_*^n}}{\epsilon}.
\end{equation}
Fatou's lemma will be used on \eqref{eq: finite diff} as $\epsilon \downarrow 0$. For this, observe that $\Delta_\epsilon (f \such f_*^n) \geq 0$ on the event $\set{f > f_*^n}$. Also, the inequality $\log(y) - \log(x) \leq (y - x) / x$, valid for $0 < x < y$, gives that, on $\set{f \leq f_*^n}$, the following lower bound holds (remember that $\epsilon \leq 1 /2$):
\[
\Delta_\epsilon (f \such f_*^n) \geq - \frac{f^n_* - f}{f^n_* - \epsilon (f^n_* - f)} \geq - \frac{f^n_* - f}{f^n_* - (f^n_* - f) / 2} = - 2 \frac{f^n_* - f}{f^n_* + f} \geq - 2.
\]
Using Fatou's Lemma on \eqref{eq: finite diff} gives $\expec \bra{(f - f^n_*) / f^n_*} \leq 0$ for all $f \in \C^n$.

Lemma \ref{lem: A1.1} again gives the existence of a sequence $(\fhat^n)_{n \in \Natural}$ such that each $\fhat^n$ is a finite convex combination of $f_*^n, f_*^{n+1},\ldots$, and $\fhat \dfn \limn \fhat^n$ exists. Since $\C$ is convex, $\fhat^n \in \C$ for all $n \in \Natural$; therefore, since $\C$ is closed, $\fhat \in \C$ as well. Fix $n \in \Natural$ and some $f \in \C^n$. For all $k \in \Natural$ with $k \geq n$, we have $f \in \C^k$. Therefore, $\expec[f / f^k_*] \leq 1$, for all $k \geq n$. Since $\fhat^n$ is a finite convex combination of $f_*^n, f_*^{n+1},\ldots$, an easy application of Jensen's inequality for the convex function $]0, \infty[ \, \ni x \mapsto 1/x \in \, ]0, \infty[$ gives that $\expec [f / \fhat^n] \leq 1$. Then, Fatou's lemma implies that for all $f \in \bigcup_{k \in \Natural} \C^k$ one has $\expec[f / \fhat] \leq 1$. The extension of the last inequality to all $f \in \C$ follows from the solidity of $\C$ by an application of the monotone convergence theorem. 
\end{proof}

By Jensen's inequality, an element $\fhat \in \C$ satisfying condition (3) of Theorem \ref{thm: main static} above is necessarily unique. (In this respect, see also \S \ref{subsubsec: numeraire}.) Therefore, the next definition makes sense.

\begin{defn}
Let $\C \subseteq \lz_+$ with $\C \cap \lz_{++} \neq \emptyset$ be convex, closed and bounded in probability. The (unique) $\fhat \in \C$ satisfying condition (3) of Theorem \ref{thm: main static} will be called the \textsl{\num \ in $\C$}.
\end{defn}

We proceed with stating and proving two results of independent interest that will help establish Proposition \ref{prop: cadlag modif}, a result concerning regularization of generalized supermartingales.

\begin{lem} \label{lem: sandwich conv}
Consider two $\lz_{+}$-valued sequences $(g^n)_{n \in \Natural}$, $(h^n)_{n \in \Natural}$ with $\expec[g^n ] \leq 1$ and $\expec[h^n] \leq 1$ for all $n \in \Natural$, as well as $\plimn (g^n  h^n) =1$. Then, $\plimn g^n = 1 = \plimn h^n$.
\end{lem}

\begin{proof}
The fact that $\plimn (g^n  h^n) =1$ implies that $\plimn \sqrt{g^n  h^n} = 1$; then
\[
\limsup_{n \to \infty}  \pare{1 - \expec \bra{\sqrt{g^n  h^n}}} = 1 - \liminf_{n \to \infty}  \expec \bra{\sqrt{g^n  h^n}} \leq 0,
\]
as follows from Fatou's Lemma. Now, since
\[
\expec \bra{ \pare{ \sqrt{g^n} - \sqrt{h^n} }^2} = \expec \bra{g^n} + \expec \bra{h^n} - 2 \expec \bra{ \sqrt{g^n h^n}} \leq 2 \pare{1 - \expec \bra{ \sqrt{g^n h^n}} },
\]
we obtain that $\plimn \pare{\sqrt{g^n} - \sqrt{h^n}} = 0$. In view of $g^n - h^n = \big(\sqrt{g^n} - \sqrt{h^n}\big)\big(\sqrt{g^n} + \sqrt{h^n}\big)$ and the fact that both sequences $(g^n)_{n \in \Natural}$, $(h^n)_{n \in \Natural}$ are bounded in probability (because $\expec[g^n] \leq 1$ and $\expec[h^n] \leq 1$ for all $n \in \Natural$), we also have $\plimn \pare{g^n - h^n} = 0$. Furthermore, the equality $g^n + h^n = \big( \sqrt{g^n} - \sqrt{h^n} \big)^2 + 2 \sqrt{g^n h^n}$ gives $\plimn \pare{g^n + h^n} = 2$. Finally, combining $\plimn \pare{g^n - h^n} = 0$ and $\plimn \pare{g^n + h^n} = 2$ gives $\plimn g^n = 1 = \plimn h^n$.
\end{proof}

\begin{prop} \label{prop: decr and incr conv of num}
For each $n \in \Natural \cup \set{\infty}$, let $\C^n$ be a convex, closed and bounded subset of $\lz_+$ with $\C^n \cap \lzpp \neq \emptyset$, and let $\fhat^n$ be the \num \ in $\C^n$. (These \num s exist in view of Theorem \ref{thm: main static}.) Then, $\plimn \fhat^n = \fhat^\infty$ holds in either of the following cases:
\begin{enumerate}
  \item $(\C^n)_{n \in \Natural}$ is nondecreasing and $\C^\infty$ is the closure in probability of $\bigcup_{n \in \Natural} \C^n$.
  \item $(\C^n)_{n \in \Natural}$ is nonincreasing and $\C^\infty = \bigcap_{n \in \Natural} \C^n$.
\end{enumerate}
\end{prop}

\begin{proof}
In the course of the proof below we drop all superscripts ``$\infty$'' to ease the readability. To establish both statements (1) and (2) below, we shall just show the existence of a subsequence $(\fhat^{m_n})_{n \in \Natural}$ of $(\fhat^{n})_{n \in \Natural}$ such that $\plimn \fhat^{m_n} = \fhat$. By the same argument, it will follow that \emph{any} subsequence of $(\fhat^{n})_{n \in \Natural}$ has a \emph{further} subsequence that converges to $\fhat$. Since $\lzp$ is equipped with a  metric topology, this will imply that the whole sequence $(\fhat^{n})_{n \in \Natural}$ converges to $\fhat$.

\smallskip
\noindent \emph{Proof of (1)}. Lemma \ref{lem: A1.1} gives the existence of a sequence $(\tf^n)_{n \in \Natural}$ such that each $\tf^n$ is a convex combination of $(\fhat^k)_{k = n, \ldots, m_n}$ for some $n \leq m_n \in \Natural$, and such that $\tf \dfn \plimn \tf^n$ exists. Of course, $\tf \in \C$. Obviously, $\limn m_n = \infty$; we can also also assume that $(m_n)_{n \in \Natural}$ is an increasing sequence, forcing it to be if necessary.

Since $\expec[f / \fhat^k] \leq 1$ holds for all $f \in \C^n$ and $n \leq k$, Jensen's inequality applied by using the convex function $]0, \infty[ \ni x \mapsto 1/x \in ]0, \infty[$ implies that $\expec[f / \tf^k] \leq 1$ holds for all $f \in \C^n$ and $n \leq k$. By Fatou's lemma, $\expec[f / \tf] \leq 1$ holds for all $n \in \Natural$ and $f \in \C^n$. In particular, $\tf \in \C \cap \lzpp$. As $(\C^n)_{n \in \Natural}$ is nondecreasing and $\C$ is the $\lz$-closure of $\bigcup_{n \in \Natural} \C^n$, Fatou's lemma applied once again will give $\expec[f / \tf] \leq 1$ for all $f \in \C$. By uniqueness of the \num, we get $\tf = \fhat$. Since $\fhat \in \lzpp$, it follows that $\plimn (\tf^n / \fhat) = 1$.

Since $\fhat^{m_n}$ is the \num  \ in $\C^{m_n}$ and $\tf^n \in \C^{m_n}$ for all $n \in \Natural$, $\expec[\tf^n / \fhat^{m_n}] \leq 1$ holds for all $n \in \Natural$. Also, $\expec[\fhat^{m_n} / \fhat] \leq 1$ is obvious because $\fhat$ is the \num \ in $\C$. Letting $g^n \dfn \tf^n / \fhat^{m_n}$ and $h^n \dfn \fhat^{m_n} / \fhat$ for all $n \in \Natural$, the conditions of the statement of Lemma \ref{lem: sandwich conv} are satisfied. Therefore, $\plimn h^n = 1$, which exactly translates to $\plimn \fhat^{m_n} = \fhat$.

\smallskip
\noindent \emph{Proof of (2)}. One applies again Lemma \ref{lem: A1.1} to get the existence of a sequence $(\tf^n)_{n \in \Natural}$ such that each $\tf^n$ is a convex combination of $(\fhat^k)_{k = n, \ldots, \ell_n}$ for some $n \leq \ell_n \in \Natural$, and such that $\tf \dfn \plimn \tf^n$ exists. We can assume that $(\ell_n)_{n \in \Natural}$ is an increasing sequence, forcing it to be if necessary. Following the same reasoning as in the proof of case (1) one can show that $\tf = \fhat$.

Define $m_0 = 1$ and a $\Natural$-valued increasing sequence $(m_n)_{n \in \Natural}$ inductively via $m_n = \ell_{m_{n-1}}$ for all $n \in \Natural$. Then, it is straightforward to check that $\expec[\fhat^{m_n} / \tf^{m_{n-1}}] \leq 1$ and $\expec[\tf^{m_{n}} / \fhat^{m_n}] \leq 1$ hold for all $n \in \Natural$. Letting $g^n \dfn \fhat^{m_n} / \tf^{m_{n-1}}$ and $h^n \dfn \tf^{m_{n}} / \fhat^{m_n}$ for all $n \in \Natural$, the conditions of the statement of Lemma \ref{lem: sandwich conv} are satisfied. Therefore, $\plimn h^n = 1$, which, in view of $\plimn \tf^{m_{n}} = \fhat$ gives $\plimn \fhat^{m_n} = \fhat$.
\end{proof}

\begin{rem}
The result of Proposition \ref{prop: decr and incr conv of num}(2) does not necessarily hold if $\C^\infty \cap \lzpp = \emptyset$. Indeed, let
  $\Omega = (0, 1]$, $\F$ be the Borel $\sigma$-field on $\Omega$ and
  $\prob$ be Lebesgue measure on $(\Omega, \F)$. Consider two
  \emph{nonincreasing} sequences $(\fhat^n)_{\nin}$ and $(g^n)_{\nin}$
  via $\fhat^n \dfn (1/2) \indic_{(0, \, 1/3]} + (1/n) \indic_{(1/3,
    \, n / (n+1)]} + \indic_{(n / (n+1), \, 1]}$ and $g^n \dfn
  \indic_{(0, \, 1/3]} + (1/ (5 n)) \indic_{(1/3, \, 1]}$. For each $\nin$, define
\[
\C^n \dfn \set{h \in \lzp \such h \leq (1 - \alpha) \fhat^n + \alpha g^n \text{ for some } \alpha \in [0,1]}.
\]
Of course, $(\C^n)_{\nin}$ is a nonincreasing sequence of sets that are convex, closed and \pbdd. In fact, $\fhat^n$ is the \num \ in $\C^n$ for all $\nin$ as it easily follows from the inequality
  \[
  \expec \bra{\frac{g^n}{\fhat^n}} = 2 \pare{\frac{1}{3}} +
  \frac{1}{5} \pare{\frac{n}{n + 1} - \frac{1}{3}} +
  \frac{1}{5n} \pare{\frac{1}{n + 1}} \leq \frac{2}{3} + \frac{2}{15}
  + \frac{1}{5} = 1.
  \]
Now, $\C^\infty =
  \bigcap_{\nin} \C^n = \{ h \indic_{(0, \, 1/3]} \such h \in \lzp \text{ with } h \leq
  1 \}$, from which it follows that $\fhat = \indic_{(0, \,
    1/3]}$ is the \num \ in $\C^\infty$. However, the sequence $(\fhat^n)_{\nin}$ converges in probability to $(1/2) \indic_{(0, 1/3]}$, which is distinct from $\fhat$.
\end{rem}

The next result concerns the ``regularization in probability'' of processes and is the analogue of path regularization of nonnegative supermartingales (see, for example, Proposition 1.3.14 of \cite{MR1121940}). Before the statement of Proposition \ref{prop: cadlag modif}, we introduce some notation. Fix a nonnegative process $X \in \X$. For $s \in [0, T[$, if $\plimn X_{t^n}$ exists and is the \emph{same} for any \emph{strictly} decreasing $[0, T]$-valued sequence $(t^n)_{n \in \Natural}$ such that $\limn t^n = s$, we shall be denoting this common limit by $\plim_{t \dda s} X_t$. By definition, we set $\plim_{t \dda T} X_t = X_T$.
Similarly, if $t \in ]0, T]$ and $\plimn X_{s^n}$ exists and is the same for any \emph{strictly} increasing $[0, T]$-valued sequence $(s^n)_{n \in \Natural}$ such that $\limn s^n = t$, we shall be denoting this latter limit by $\plim_{s \dua t} X_s$.

\begin{prop} \label{prop: cadlag modif}
Let $Z$ be a strictly positive generalized supermartingale with respect to $\bF$. Then, for all $t \in [0, T]$, $Z_{t+} \dfn \plim_{\tau \dda t} Z_\tau$ exists. If $\tau \in \, ]0, T]$, $Z_{\tau-} \dfn \plim_{t \dua \tau} Z_t$ exists as well. Furthermore, $(Z_{t+})_{t \in [0, T]}$ is a strictly positive generalized supermartingale with respect to $\bF$, and $\plim_{t \dua \tau} Z_{t+}$ exists and is equal to $Z_{\tau -}$ for all $\tau \in \, ]0, T]$.
\end{prop}

\begin{proof}
For $t \in [0, T]$, let $\C_t$ be the closed (in probability) convex hull of $\{ Z_\tau \such \tau \in [t, T] \}$. It follows that $\C_t \subseteq \C_s$ whenever $s \in [0, T]$ and $t \in [s, T]$. Also, $Z_t$ is the \num \ in $\C_t$, since $\expec[Z_\tau / Z_t ] \leq 1$ whenever $t \in [0, T]$ and $\tau \in [t, T]$. In particular, in view of Theorem \ref{thm: main static}, $\C_t$ is \pbdd \ for all $t \in [0, T]$.

For all $t \in [0, T[$, let $\C_{t +} \dfn \bigcup_{\tau \in \, \otcT } \C_\tau$, as well as $\C_{T+} \dfn \C_T$. For all $t \in [0, T]$, $\C_{t+} \subseteq \C_{t}$, and $\C_{t+} =  \bigcup_{n \in \Natural} \C_{\tau^n}$ holds for any \emph{strictly} decreasing $[0, T]$-valued sequence $(\tau^n)_{n \in \Natural}$ with $\limn \tau^n = t$ whenever $t \in [0, T[$. An application of Proposition \ref{prop: decr and incr conv of num} gives that $Z_{t+} \dfn \plim_{\tau \dda t} Z_\tau$ exists for all $t \in [0, T]$ and it is actually equal to the \num \ in $\oC_{t+}$, where $\oC_{t+}$ will denote the closure in probability of $\C_{t+}$. (Observe that the \num \ in $\oC_{t+}$ exists by Theorem \ref{thm: main static}, as $\C_{t+} \cap \lz_{++} \neq \emptyset$ and $\oC_{t+}$ is convex and \pbdd.)

Consider now the process $Z_{\cdot +} \dfn (Z_{t+})_{t \in [0, T]}$. Since $\C_{t+} \cap \lz_{++} \neq \emptyset$ for all $t \in [0, T]$ and $Z_{t+}$ is the \num \ in $\oC_{t+}$, it follows that $Z_{t+} \in \lz_{++}$, i.e., $Z_{\cdot +}$ is strictly positive. We claim that $Z_{\cdot +}$ is \lzcadlag; indeed, for $t \in [0, T[$, and as $\oC_{t +}$ coincides with the closure in probability of $\bigcup_{\tau \in \, \otcT} \oC_{\tau +}$, an application of Proposition \ref{prop: decr and incr conv of num}(1) gives that $Z_{t +} = \plim_{\tau \dda t} Z_{\tau+}$. Now, for all $\tau \in \, ]0, T]$ we have $\bigcap_{t \in [0, \tau[ } \oC_{t+} = \bigcap_{t \in [0, \tau[} \C_{t}$. An application of Proposition \ref{prop: decr and incr conv of num}(2) gives that $\plim_{t \dua \tau} Z_{t+}$ and $\plim_{t \dua \tau} Z_{t}$ exist, and they are actually equal.

It only remains to show that $\expec[Z_{t+} / Z_{s+} \such \F_s] \leq 1$ holds whenever $s \in [0, T]$ and $t \in [s, T]$. Fix $s \in [0, T]$ and $t \in [s, T]$, as well as $A \in \F_s$. For all $n \in \Natural$, with $s^n \dfn (1 - 1/n)s + T/n$ and $t^n \dfn (1 - 1/n)t + T/n$, the generalized supermartingale property of $Z$ with respect to $\bF$ and the fact that $A \in \F_s \subseteq \F_{s^n}$ give $\expec[(Z_{t^n} / Z_{s^n}) \indic_A] \leq \prob[A]$. Then, Fatou's lemma gives $\expec[(Z_{t+} / Z_{s+} ) \indic_A] \leq \prob[A]$. Since $A \in \F_s$ was arbitrary we get $\expec[Z_{t+} / Z_{s+} \such \F_s ] \leq 1$.
\end{proof}

\begin{rem}

The proof of Proposition \ref{prop: cadlag modif} is based on a generalization of Doob's celebrated result regarding the convergence of nonnegative supermartingales. For simplicity, we discuss the case where the time-set is discrete, i.e., the process is indexed by $\Natural$ --- the extension to $\Real_+$-indexed processes is straightforward. Let $(g_n)_{\nin}$ be an $\lzp$-valued sequence of random variables such that $\expec[g_n / g_m] \leq 1$ holds whenever $\Natural \ni m \leq n \in \Natural$, and such that the convex hull of $\set{g_n \such \nin}$ is bounded away from zero in probability. Following the ideas in the proof of Proposition \ref{prop: cadlag modif} --- more precisely, using statement (2) of Proposition \ref{prop: decr and incr conv of num} --- we can obtain that $\plimn g_n$ exists. To compare this result with the nonnegative supermartingale convergence theorem, let $\Hcal_n$ denote the smallest $\sigma$-field that makes all random variables $g_1, \ldots, g_n$ measurable. Doob's well-known result states that if $\expec[g_n  \such \Hcal_m] \leq g_m$ holds whenever $\Natural \ni m \leq n \in \Natural$, then $\limn g_n$ $\prob$-a.s. exists. Rewrite the supermartingale property $\expec[g_n  \such \Hcal_m] \leq g_m$ as $\expec[g_n / g_m \such \Hcal_m] \leq 1$, and note in particular that $\expec[g_n / g_m] \leq 1$ whenever $\Natural \ni m \leq n \in \Natural$. (This is just the generalized supermartingale property of $(g_n)_{\nin}$ under the trivial filtration.) Therefore, the supermartingale convergence theorem becomes a special case of our result, since no conditioning is used in the generalized supermartingale property of $(g_n)_{\nin}$. However, one can no longer claim that $(g_n)_{\nin}$ converges $\prob$-a.s.; this is the reason why only a regularization ``in probability'' is obtained in Proposition \ref{prop: cadlag modif}.

We remark that this generalization of Doob's supermartingale convergence theorem seems new (the author was not able to spot any such occurrence in the literature) and its proof does not use ``traditional'' methods and tools of martingale theory.
\end{rem}

We are now ready to give the proof of Theorem \ref{thm: main dynamic}.

\begin{proof}[Proof of Theorem \ref{thm: main dynamic}.] We show the implications $(1) \Rightarrow (2)$, $(1) \Rightarrow (3)$, $(3) \Rightarrow (2)$ and $(2) \Rightarrow (1)$ below. The fact that all processes in $\X$ are \lzcadlag \ under any of the equivalent conditions (1) or (2) in Theorem \ref{thm: main dynamic} is discussed after the proof of implication $(2) \Rightarrow (1)$. As discussed in \S \ref{subsubsec: financial interpretation}, we can, and shall, assume that property (b) of the set $\X$ in the statement of Theorem \ref{thm: main dynamic} is strengthened into $1 \in \X$.

\smallskip

\noindent \underline{$(1) \Rightarrow (2)$}. For all $t \in [0, T]$, let $\C_t \dfn \set{X_t \such X \in \X}$. The convexity of $\X$ implies that $\C_t$ is convex for all $t \in [0, T]$. Let $X \in \X$. The switching property of $\X$, combined with $1 \in \X$ gives that $\tX \dfn X_{t \wedge \cdot}$ is also in $\X$; since $\tX_T = X_t$, we obtain that $\set{X_t \such X \in \X} \subseteq \set{X_T \such X \in \X}$. Therefore, $\C_t$ is \pbdd \ for all $t \in [0, T]$. From Theorem \ref{thm: main static} it follows that, for all $t \in [0, T]$, there exists $\fhat_t$ in the closure in probability of $\C_t$ such that $\expec[f / \fhat_t] \leq 1$ holds for all $f \in \C_t$.

Now, let $(\xi^n)_{n \in \Natural}$ be a sequence in $\X$ such that $\xi^n_T \in \lzpp$ for all $\nin$ and $\plimn \xi^n_T = \fhat_T$. We shall show that $\plimn \xi^n_t = \fhat_t$ actually holds for all $t \in [0, T]$. Fix $t \in [0, T]$ and let $(\chi^n)_{n \in \Natural}$ be a sequence in $\X$ such that $\chi^n_t \in \lzpp$ for all $\nin$ and $\plimn \chi^n_t = \fhat_t$. We can assume without loss of generality that $\expec[\xi^n_t / \chi^n_t] \leq 1$ for all $n \in \Natural$. (Indeed, if the latter fails we can replace $\chi^n$ with $\psi^n$, an appropriate convex combination of $\chi^n$ and $\xi^n$, such that $\expec[\xi^n_t / \psi^n_t] \leq 1$ and $\expec[\chi^n_t / \psi^n_t] \leq 1$ hold for all $n \in \Natural$; in effect, $\psi^n_t$ is the \num \ in $\set{(1 - \alpha) \chi^n_t + \alpha \xi^n_t \such \alpha \in [0,1]}$. Lemma \ref{lem: sandwich conv} with $g^n \dfn \chi_t^n / \psi_t^n$ and $h^n = \psi_t^n / \fhat_t$ for all $n \in \Natural$ implies that this new $\C_t$-valued sequence $(\psi^n_t)_{n \in \Natural}$ will still converge to $\fhat_t$.) Now, for each $n \in \Natural$, let $\zeta^n \dfn \chi^n_{t \wedge \cdot} (\xi^n_{t \vee \cdot} / \xi^n_t)$. We have $\zeta^n \in \X$ by the switching property, and $\zeta^n_T = (\chi^n_t / \xi^n_t) \xi^n_{T}$. Then, $\expec[\xi^n_T / \zeta^n_T ] = \expec[\xi^n_t / \chi^n_t ] \leq 1$ for all $n \in \Natural$. An application of Lemma \ref{lem: sandwich conv} with $g^n \dfn \xi^n_T / \zeta^n_T$ and $f^n \dfn \zeta^n_T / \fhat_T$ for $n \in \Natural$ gives $\plimn \zeta^n_T = \fhat_T$. Combining this with $\plimn \chi^n_t = \fhat_t$, we get $\plimn (\xi_t^n / \xi_T^n) = \fhat_t / \fhat_T$, and, therefore, $\plimn \xi_t^n = \fhat_t$, which is the claim we wished to establish.

Define $\hY_t \dfn 1 / \fhat_t$ for all $t \in [0, T]$; as $\fhat_t \in \lz_{++}$, $\hY$ is a well-defined and strictly positive process. We claim that $\limn \expec \big[ |\hY_t \xi^n_t - 1 | \big] = 0$ holds for each $t \in [0, T]$. Indeed, since $\plimn (\hY_t \xi^n_t) = 1$ and $(\hY_t \xi^n_t) \in \lz_+$ for all $n \in \Natural$,
by Theorem 16.14(ii), page 217 in \cite{MR1155402} one needs to establish that $\limn \expec[\hY_t \xi^n_t] = 1$, which follows from $1 = \expec \big[ \liminf_{n \to \infty} \hY_t \xi^n_t \big] \leq \liminf_{n \to \infty}  \expec[ \hY_t \xi^n_t] \leq \limsup_{n \to \infty}  \expec[ \hY_t \xi^n_t] \leq 1$. In particular, for all $A \in \G$ we have $\limn \expec[\hY_t \xi^n_t \indic_A] = \prob[A]$.

Fix $s \in [0, T]$, $t \in [s, T]$, $A \in \F_s$ and a strictly positive $X \in \X$. For $n \in \Natural$, let $\tX^n \dfn \indic_{\Omega \setminus A} \, \xi^n_{\cdot} + \indic_{A} (\xi^n_{s} / X_{s} ) X_{s \vee \cdot}$. The switching property of $\X$ implies that $\tX^n \in \X$. Furthermore, $\tX^n_t = \indic_{\Omega \setminus A} \, \xi^n_{t} + \indic_{A} (\xi^n_{s} / X_{s} ) X_{t}$. Then, $\expec[\tX^n_t \hY_t] \leq 1$ translates to the inequality $\expec \big[ (X_t / X_s) \hY_t \xi^n_s \indic_A \big] \leq 1 - \expec[\indic_{\Omega \setminus A} \hY_t \xi^n_t]$.
Using Fatou's lemma on the left-hand side of this inequality and the fact that $\limn \expec[\indic_{\Omega \setminus A} \hY_t \xi^n_t] = 1 - \prob[A]$ on the right-hand-side, we obtain
\begin{equation} \label{eq: pre-deflator property}
\expec \bra{\frac{X_t \hY_t}{X_s \hY_s} \indic_A} \leq \prob[A].
\end{equation}
Since $A \in \F_s$ was arbitrary, it follows that $\expec \big[ X_t \hY_t / (X_s \hY_s) \such \F_s \big] \leq 1$ for all strictly positive $X \in \X$. 

Since $1 \in \X$, using Proposition  \ref{prop: cadlag modif} with $Z \dfn \hY$ we obtain a strictly positive generalized supermartingale $Y$ with respect to $\bF$, such that $Y_0 = 1$ and $Y_t = \plim_{\tau \dda t} \hY_{\tau}$ holds for all $t \in [0, T]$. Fix $s \in [0, T]$, $t \in [s, T]$, $A \in \F_s$ and a strictly positive $X \in \X$. For all $n \in \Natural$, let $s^n \dfn (1 - 1/n)s + T/n$ and $t^n \dfn (1 - 1/n)t + T/n$. For all $n \in \Natural$, and since $A \in \F_s$, we have $\expec[(\hY_{t^n} X_{t^n} / (\hY_{s^n} X_{s^n})) \indic_A] \leq \prob[A]$ by \eqref{eq: pre-deflator property}. As $X$ is \lzcad, Fatou's lemma gives $\expec[(Y_{t} X_{t} / (Y_{s} X_{s})) \indic_A] \leq \prob[A]$. Since $A \in \F_s$ was arbitrary we obtain $\expec[Y_{t} X_{t} / (Y_{s} X_{s}) \such \F_s ] \leq 1$ for all strictly positive $X \in \X$. We have to show that the last inequality actually holds also for all $X \in \X$, not necessarily strictly positive. Fix then $X \in \X$ and let $X^n \dfn (1/n) + (1 - 1/n) X$ for all $n \in \Natural$; then, $X^n \in \X$ and $X^n$ is strictly positive. It follows that $\expec[Y_{t} X^n_{t} / (Y_{s} X^n_{s}) \such \F_s] \leq 1$ for all $n \in \Natural$. Now, $\liminf_{n \to \infty} (X^n_t / X^n_s) = (X_t / X_s) \indic_{\set{X_s > 0}} + \indic_{\set{X_s = 0, \, X_t = 0}} + \infty \indic_{\set{X_s = 0, \, X_t > 0}}$. As $\expec[ \liminf_{n \to \infty} (Y_{t} X^n_{t} / (Y_{s} X^n_{s})) \such \F_s] \leq 1$ holds by the conditional version of Fatou's lemma, and $\prob[Y_s > 0, \, Y_t > 0] = 1$, we obtain $\prob[X_s = 0, \, X_t > 0] = 0$. Then, using the division conventions mentioned in \S \ref{subsuc: gen supermarts}, we get $\expec[Y_{t} X_{t} / (Y_{s} X_{s}) \such \F_s ] \leq 1$ for all $X \in \X$. In other words, $Y X$ is a nonnegative generalized supermartingale with respect to $\bF$ for all $X \in \X$.

\smallskip

\noindent \underline{$(1) \Rightarrow (3)$}. The implication $(1) \Rightarrow (3)$ of Theorem \ref{thm: main static}, applied to the set $\C \dfn \set{X_T \such X \in \X}$ (which is assumed closed) implies that there exists $\hX \in \X$ such that $\expec[X_T / \hX_T] \leq 1$ for all $X \in \X$. We shall show that $X / \hX$ is a nonnegative generalized supermartingale with respect to $\bF$ for all $X \in \X$. The proof of implication $(1) \Rightarrow (2)$ above shows that $\expec[X_t / \hX_t] \leq 1$ for all $X \in \X$ and $t \in [0, T]$; in particular, $\hX$ is strictly positive. Using the notation of the proof implication $(1) \Rightarrow (2)$, it is clear that $\hX = 1 / \hY$. Then, the result follows directly from \eqref{eq: pre-deflator property}.

\smallskip

\noindent \underline{$(3) \Rightarrow (2)$}. Assume that $\hX$ exists and set $\hY \dfn 1 / \hX$. A priori, $\hY$ is not necessarily \lzcadlag. However, passing to $Y$ as in the proof of implication $(1) \Rightarrow (2)$ above and following the rest of the argument, we can conclude the existence of a generalized supermartingale deflator. 

\smallskip

\noindent \underline{$(2) \Rightarrow (1)$}. Pick $Y$ with the properties of statement (2). For all $\ell \in \Real_+$, we have the inequality $\ell \sup_{X \in \X} \prob[Y_T X_T > \ell] \leq \sup_{X \in \X} \expec[Y_T X_T] \leq  1$. Therefore, the set $\set{Y_T X_T \such X \in \X}$ is \pbdd. Since $Y_T \in \lzpp$, $\set{X_T \such X \in \X}$ is \pbdd.

\medskip

Finally, we establish that if $Y$ is a process satisfying condition (2) of Theorem \ref{thm: main dynamic}, all wealth processes in $\X$ are \lzcadlag. Pick $X \in \X$. Let $X' = (1 + X)/2$; then $X' \in \X$ and $X'$ is strictly positive. It follows that $Y X'$ is a strictly positive generalized supermartingale with respect to $\bF$. According to Proposition \ref{prop: cadlag modif}, $\plim_{t \dua \tau} (Y_t X'_t)$ exists for all $\tau \in \, ]0, T]$; as $\plim_{t \dua \tau} Y_t$ also exists and is an element of $\lzpp$, we obtain that $\plim_{t \dua \tau} X'_t$ exists for all $\tau \in\, ]0, T]$. This is equivalent to saying that $\plim_{t \dua \tau} X_t$ exists for all $\tau \in \, ]0, T]$. Since $X$ is already \lzcad, we conclude that $X$ is \lzcadlag.
\end{proof}

\bibliographystyle{siam}
\bibliography{gsd}
\end{document}